 \documentclass[12pt, draftclsnofoot, onecolumn]{IEEEtran}
\ifCLASSINFOpdf
\else
\fi
\usepackage{amsmath,amsthm}
\theoremstyle{plain}

\newtheorem{thm}{\normalfont \textit{Theorem}}[section]

\usepackage{xpatch}
\makeatletter
\xpatchcmd{\@thm}{\thm@headpunct{.}}{\thm@headpunct{}}{}{}
\makeatother
\makeatletter
\renewenvironment{proof}[1][\proofname]{\par
\pushQED{\qed}%
\normalfont 
\trivlist
\item\relax
{\itshape
#1\@addpunct{:}}\hspace\labelsep\ignorespaces
}{%
\popQED\endtrivlist\@endpefalse
}

\makeatother

\usepackage{bm,amsmath,amssymb,latexsym,float,epsfig,subfig,epstopdf,graphicx,array,algorithm,algpseudocode,makecell,mathtools}
\usepackage{tikz}
\usepackage{setspace}
\usetikzlibrary{positioning,shapes,arrows,fit,calc,chains}
\usepackage{footmisc}

\newcommand{\symbolB}{
}

\begin{document}

\tikzstyle{block} = [draw, fill=blue!20, rectangle,
    minimum height=2.5em, minimum width=2.5em]
\tikzstyle{sum} = [draw, fill=blue!20, circle, node distance=1cm]
\tikzstyle{input} = [coordinate]
\tikzstyle{output} = [coordinate]
\tikzstyle{pinstyle} = [pin edge={to-,thin,black}]

%
\title{Precoded Chebyshev-NLMS based pre-distorter for nonlinear LED compensation in NOMA-VLC}
%
%
%

\author{Rangeet~Mitra,
        Vimal~Bhatia,~\IEEEmembership{Senior Member,~IEEE}
\thanks{R. Mitra and V. Bhatia are with Indian Institute of Technology Indore, Indore-453552, India, Email:phd1301202010@iiti.ac.in, vbhatia@iiti.ac.in. This work was submitted to IEEE Transactions on Communications on $26^{\text{th}}$ October, 2016, decisioned on $3^{\text{rd}}$ March, 2017, and revised on $25^{\text{th}}$ April, 2017, and is currently under review in IEEE Transactions on Communications.}}
\maketitle

\begin{abstract}
Visible light communication (VLC) is one of the main technologies driving the future 5G communication systems due to its ability to support high data rates with low power consumption, thereby facilitating high speed green communications. To further increase the capacity of VLC systems, a technique called non-orthogonal multiple access (NOMA) has been suggested to cater to increasing demand for bandwidth, whereby users' signals are superimposed prior to transmission and detected at each user equipment using successive interference cancellation (SIC).
Some recent results on NOMA exist which greatly enhance the achievable capacity as compared to orthogonal multiple access techniques. However, one of the performance-limiting factors affecting VLC systems is the nonlinear characteristics of a light emitting diode (LED). This paper considers the nonlinear LED characteristics in the design of pre-distorter for cognitive radio inspired NOMA in VLC, and proposes singular value decomposition based Chebyshev precoding to improve performance of nonlinear multiple-input multiple output NOMA-VLC. 
A novel and generalized power allocation strategy is also derived in this work, which is valid even in scenarios when users experience similar channels. Additionally, in this work, analytical upper bounds for the bit error rate  of the proposed detector are derived for square $M$-quadrature amplitude modulation.
\end{abstract}


%
\IEEEpeerreviewmaketitle

\section{Introduction}\label{Sec1}
Visible light communications (VLC) \cite{inan2009impact} is a viable supplement to existing radio frequency (RF) communications, since shifting communication towards the nm-wave range \cite{haas_what_is} results in expansion of the available bandwidth to meet the spectrum and low power demands for the proposed 5G systems. In VLC, the light emitting diode (LED)-lamps, generally used as luminaries, additionally serve as optical-transmitters, and photodiode-arrays are used as optical receivers in a typical optical attocell \cite{yin2016performance}. The intensity modulation of LED by users' signals is done at a speed imperceptible to human eye \cite{inan2009impact}, thereby achieving the dual goal of illumination and signal transmission at lower power as compared to RF equipment. VLC is a core component of Li-Fi systems \cite{haas_what_is}, which has been proposed for applications like Internet of Things (IoT), 5G systems, underwater communications, vehicle-to-vehicle communication and many others.

In parallel, to meet the ever increasing data demand of users in VLC for 5G systems (by 2020), a novel multiple access scheme called non-orthogonal multiple access (NOMA) (or more specifically power division multiple access) has attracted  attention \cite{saito2013non,dai2015non}. In this multiple access scheme, instead of allocating orthogonal time-frequency (TF) resources, the signals of multiple users are overlapped over the same TF resource. The users are allocated distinct power levels depending on their corresponding channel condition. Users with better channel conditions are allocated less power, while the users with poorer channel conditions are allocated more power. Consequently, the symbol detection at each user equipment (UE) is performed by using a successive interference canceller (SIC). Recently, NOMA is found to be viable for adoption into VLC \cite{hanaa} mainly due to the following reasons:
a) NOMA can handle small number of users which typically happens in a Li-Fi attocell (a small femtocell) \cite{haas_what_is}, b) channel is generally dominated by a line of sight (LOS) path which facilitates for accurate channel estimation, and c) adjusting tuning angles and field of view (FOV) gives additional degrees of freedom for multiplexing multi-user signals which can induce differential channel gains facilitating for power diversity.

Amongst existing techniques for NOMA precoding, the work in \cite{ding2016application} proposes different precoding matrices which are based on the assumption of different channel conditions at each UE. In \cite{ding2016application}, the authors propose a maximal ratio combining (MRC) based post-processing, which would work well if the channel matrices are well conditioned/ the columns of the channel matrices of all users (with all users having different quality of service (QoS) requirements) are all independent/dissimilar to ensure power-diversity. However, if the channel matrices are aligned/similar, the MRC based processing proposed in \cite{ding2016application} would render similar channel conditions at each UE thereby culminating in failure of the NOMA system. To avoid similar channels, pre-processing at the transmitter is considered in  \cite{hanif2016minorization} (is related to the problem addressed in this paper), however it assumes a single antenna for each UE in its formulation. In multiple-input multiple output (MIMO)-VLC, however, we have an array of photodiodes at each UE thus restricting the validity of the approach in \cite{hanif2016minorization} (which essentially decomposes the MIMO detection problem into several spaced multiple input single output (MISO) problems thereby losing diversity gain \cite{hanif2016minorization}). The work in \cite{hanif2016minorization} could be considered as ``one-antenna per UE"-analogue of the system considered in this paper. 

In IoT applications, the existence of different channel conditions for all UEs is not guaranteed. In such applications, channel diversity can be achieved by using precoding. The precoding technique suggested in \cite{ding} holds even when channel conditions for UEs are similar, however is applicable to only two-user scenario, and leaves its possible extension to arbitrary number of users as an open problem.
The works in \cite{ding2016application}, and \cite{ding} proposes a new paradigm for NOMA power allocation called ``cognitive-radio (CR) inspired power allocation" (also called CR-NOMA). In this paradigm, the user with a better channel condition is analogous to a secondary-user, and the user with worse channel condition is analogous to a primary-user. When a new user wants to access the link, it is served opportunistically under the condition that the existing users' QoS requirements are maintained. This policy applies in many 5G scenarios like IoT where users have diverse QoS requirements and not necessarily diverse channel conditions \cite{ding}. To accommodate multiplicity of users in such scenarios, it is crucial to have diverse channel conditions which can be either in-built as in \cite{ding2016application}, or achieved by precoding \cite{ding}. In this work, techniques for CR inspired NOMA downlink (which arises typically in 5G scenarios like IoT \cite{ding}) are proposed wherein channels of all users exhibit significant correlation in addition to device-impairments.

Despite employing NOMA-VLC as detailed above, the gains expected from NOMA-VLC systems is limited by inherent LED nonlinearity \cite{lee1977nonlinearity}. This LED nonlinearity can be mitigated by using pre-distortion or post-distortion techniques. Among pre-distortion techniques, the simplest technique would be to maintain a lookup table of the estimated nonlinearity. However, such static lookup tables lose their practicability due to
varying LED  characteristics \cite{kim2014adaptive} (due to device aging); hence there is a need for adaptive pre-distorters. Such adaptive pre-distorters have been suggested in the literature which are learnt using normalized least mean squares (NLMS) algorithm \cite{kim2014adaptive}, and Chebyshev regression using NLMS approach \cite{mitra_chebyshev}.

Apart from the requirement of channel diversity at each UE, and the need for mitigation of device impairments for each user, it is also essential to incorporate the QoS requirements of each user \cite{ding} (specified as design parameters in terms of rate) into the design of the precoder, power-allocation strategy, and the BER analysis of the overall system. In this regard,  some recent works in NOMA have the following limitations: a) the work in \cite{yang2017fair} considers a single input single output (SISO) channel and does not consider the individual users' QoS requirements, b) the work in \cite{zhang2016user} considers satisfying an individual users' QoS requirements in a SISO channel. Since the Lebesgue measure of the users's QoS constraint sets becomes small or even zero in a correlated channel scenario, hence \cite{zhang2016user} cannot be generalized for correlated MIMO channels, and c) although \cite{yin2016performance}, and \cite{hanaa} consider and analyze NOMA systems by considering user mobility, the case when the users experience similar and static channels (along with the problems being framed in the SISO setting as opposed to MIMO-NOMA) is not considered in both \cite{yin2016performance} and \cite{hanaa}. 

Additionally, to the best of authors' knowledge, all studies so far in a VLC-NOMA system \cite{hanaa,NOMA1,NOMA2} do not consider LED nonlinearity for the design of a pre-distorter. In this paper, a modified Chebyshev-NLMS based pre-distortion is proposed with hybrid eigen-decomposition based precoding in a MIMO NOMA-VLC scenario for IoT applications. In view of the existing NOMA literature in VLC, the major contributions of this work are summarized as follows:
\begin{itemize}
 \item We suggest a new hybrid precoding technique for downlink NOMA-VLC channels using a closed loop adaptive Chebyshev pre-distorter, based on singular value decomposition (SVD), in IoT applications. The proposed approach works even when the left and the right eigenvectors of all the available channel matrices for all users are correlated. In other words, if the channels have correlated channel state information (CSI), one cannot design different precoding vectors for users by the precoding framework proposed in \cite{ding2016application}, and \cite{marshoud2015mu}. For two-user NOMA, a QR decomposition based technique, as proposed in \cite{ding}, is able to deliver varying grades of QoS even in scenarios with similar channel conditions. However, the precoding technique presented in this paper holds for arbitrary number of users, which is an open problem raised in \cite{ding}.
\item For the proposed precoding technique for downlink MIMO NOMA-VLC, proof of sum-rate-gain maximization upon addition of a new user for the CR-NOMA is provided in this paper. Further, the theoretical analysis in this paper considers residual estimation error, both in the adaptive pre-distorter and in the CSI-estimation.
\item A novel power allocation technique is found for the proposed precoding scheme based on QoS requirements of individual users while considering the residual nonlinearity after post-distortion, and channel estimation error. While finding the optimal power allocation coefficients, we consider the SVD based precoding technique given in this paper in a scenario when the channel matrices of all the users have similar CSI with each user being opportunistically served whilst maintaining the QoS requirements of the existing users. In such conditions as given in \cite{ding}, the users have diverse QoS requirements as opposed to diverse channel conditions. Existing work on power allocation like gain ratio power allocation (GRPA) as given in \cite{hanaa} could provide a solution, however it does not take into account an individual users' QoS, and hence fails when the users have similar channels as highlighted in this work. 
Simulations indicate that the proposed power allocation technique performs well for square $M$-quadrature amplitude modulation (QAM) modulation schemes in the above mentioned scenario for varying number of users.
\item Analytical upper bounds for bit error rate (BER) vs signal to noise ratio (SNR) for varying number of users are derived for the proposed MIMO NOMA-VLC scenario by considering estimation error, and validated by simulations for square $M$-QAM. The simulations indicate that the theoretically derived BER formulae indeed match the simulated BER curves for varying number of users, which further validates the analysis presented in this work.
\end{itemize}
In this paper, we adopt the following terminology: scalars at time $k$ are denoted by subscript $k$ such as $x_{k}$, vectors (which are tuples of scalars), are denoted by small boldface as $\textbf{x}_{k}$, and matrices are denoted by capital boldface such as $\textbf{H}$. Transpose of matrices/vectors are denoted by $(\cdot)^{T}$. Additionally, inverse of transpose of a matrix is denoted by $(\cdot)^{-T}$ and the pseudo-inverse is denoted by $(\cdot)^{\dagger}$. The statistical expectation operator is denoted by $\mathbb{E}[\cdot]$. Sets are denoted by $\{\cdot\}$ in this work, $\|\cdot\|_{q}$ denotes the $l_{q}$ norm, and $\hat{(\cdot)}$ denotes estimate of the random variable.

This paper is organized as follows: Section-\ref{Sec2} discusses the MIMO-NOMA system model, Section-\ref{Sec3} reviews NLMS and Chebyshev-NLMS based pre-distortion, and existing two-user QR precoded linear NOMA is reviewed in Section-\ref{Sec4}. Section-\ref{Sec5}, presents proof of the feasibility of the proposed precoded Chebyshev-pre-distortion for nonlinear LED affected VLC-NOMA, and Section-\ref{Sec6} suggests the choice of suitable precoding matrix. Section-\ref{Sec7} presents the power allocation strategy for the given system. Expression for BER performance for square $M$-QAM is derived in Section-\ref{Sec8}. The simulation setup, parameters and results are described in Section-\ref{Sec9}. Finally, conclusions are drawn in Section-\ref{Sec10}.
\section{mimo-noma system model}\label{Sec2}
In this section, we present the MIMO-NOMA system model considered in this paper and introduce the terminology followed
throughout the paper. The proposed system model is given in Fig. \ref{Fig_2}. We denote the input vector at $k^{th}$ time instant as, $\textbf{x}_{k}=[x_{k}]_{k=zM_{T}+1}^{(z+1)M_{T}}$, where $z$ is an arbitrary integer that denotes the sample duration, and $M_{T}$ is the number of LEDs at the  transmitter (Tx). The input vector is a non-orthogonal superposition (superposition coding) of many users' signals with suitable power allocation. Mathematically, this can be written as:
\begin{equation}
\textbf{x}_{k} = \sum_{u=1}^{U}\sqrt{P^{(u)}}\textbf{s}^{(u)}_{k}
\end{equation}
where $U$ represents total number of users and $u$ denotes the $u^{th}$ user index variable. $P^{(u)}$ is the power allocated to the $u^{th}$ user, and $\textbf{s}^{(u)}_{k}$ is the tuple of $u^{th}$ user's symbol at $k^{th}$ time instant transmitted from all LEDs. Without loss of generality, to impose a constant power constraint, it is assumed that $\sum_{u=1}^{U}P^{(u)}=1$.
The symbol $\textbf{x}_{k}^{'}$ indicates precoding of $\textbf{x}_{k}$ by an estimated precoding matrix $\hat{\textbf{P}}^{(u)}$ for the $u^{th}$ user. 
To guarantee non-negativity of the precoded vector, the biasing factor for $\textbf{x}_{k}^{'}$ can be given as follows \cite{ma2015coordinated}:
\begin{align}
\textbf{x}_{k}^{'} = &\hat{\textbf{P}}^{(u)}\textbf{x}_{k} + \max_{1<q<M_{R}}\|\tilde{\textbf{p}}_{q}^{(u)^T}\|_{1}(1+\sqrt{-1}) 
\end{align}
The $\sqrt{-1}$ operator represents modulation by an orthogonal pulse in baseband, and $\tilde{\textbf{p}}_{q}^{(u)^T}$ represents the transpose of the $q^{th}$ row of $\hat{\textbf{P}}^{(u)}$.

This is followed by pre-distorter mapping $T(.)$. The considered channel matrix for the $u^{th}$ user is denoted by $\textbf{H}^{(u)}\in\mathbb{R}^{M_{R}\times M_{T}}$, and consequently the independent and identically distributed (i.i.d) additive white Gaussian noise (AWGN) is added, and the superposition of signals is broadcasted from the transmitter LED array.
For the $u^{th}$ user, the received signal vector $\textbf{y}_{k}^{(u)}$ can be written as:
\begin{equation}
\textbf{y}_{k}^{(u)} =  \textbf{H}^{(u)}A(T(\hat{\textbf{P}}^{(u)}\textbf{x}_{k})) + \textbf{n}_{k}
\end{equation}
where $\textbf{n}_{k}$ denotes AWGN with zero mean and covariance matrix $\sigma_{n}^{2}\textbf{I}$ ($\textbf{I}$ being the identity matrix). $T(.)=\sum_{\forall i} r_{k}^{(i)}T_{i}(.)$ denotes the adaptive pre-distorter transformation which can be learnt by popular techniques like NLMS \cite{kim2014adaptive}, or Chebyshev regression based NLMS \cite{mitra_chebyshev}. $T_{i}$ denotes the $i^{th}$ Chebyshev polynomial and $r_{k}^{(i)}$ denotes the pre-distorter weights (also called/implemented as random access memory (RAM) in the literature \cite{kim2014adaptive}). Chebyshev polynomials are min-max error optimal over the unit interval, and are therefore chosen as a set of basis functions for learning the pre-distorter as follows \cite{abramowitz1964handbook}: 
\begin{gather}
T_{0}(x)=1,\\ \nonumber
T_{1}(x)=x,\\ \nonumber
T_{i+1}(x) = 2xT_{i}(x)-T_{i-1}(x)
\end{gather}
At each UE, $\{{s}_{k}^{(u)}\}$ is then recovered by SIC given in \cite{ding2016application} and \cite{ding}, after multiplication by $(\hat{\textbf{H}}^{(u)}\hat{\textbf{P}}^{(u)})^{\dagger}$. The LED nonlinearity $A(.)$ is modeled by the following equation:
\begin{equation}
A(x) = \frac{x}{(1+(\frac{x}{I_{max}})^{2p})^{\frac{1}{2p}}}
\end{equation}
where the parameter $p$ controls severity of the nonlinearity and $I_{max}$ is the maximum saturation current of the LED \cite{elgala2010led}. Typically $p=0.5$ is chosen as in \cite{elgala2010led} to model a severe nonlinearity and the same is chosen in this work.

In the case of imperfect CSI (both transmit and receive) considered in this paper, we have the following measure of deviation expressed in terms of actual elements of channel matrix $\textbf{H}$, and its corresponding estimate, $\hat{\textbf{{H}}}$, which is the channel matrix affected by estimation error \cite{ma2015coordinated}:
\begin{gather}
\|\textbf{H}-\hat{\textbf{H}}\|_{2}\leq \gamma
\end{gather}
where $\gamma$-neighborhood is assumed to be Gaussian distributed with variance $\sigma_{\gamma}^{2}$ (as given in \cite{yang2016performance,ying2015joint}), and $\|\cdot\|_{2}$ denotes the Euclidean norm. It is assumed in the following analysis, that all the users' channels are static or quasi-static, so that their CSI is not outdated till the subsequent channel estimation.



In this paragraph, we describe every considered block of the system model from Fig. \ref{Fig_2}. In the first block the users' respective signals are superposition coded resulting in input $\textbf{x}_{k}$, which is precoded by the matrix $\hat{\textbf{P}}^{(u)}$ corresponding to user $u$, and passed through the pre-distorter $T(.)$. This precoded transmission is  passed to the transmitter LED array, where the nonlinearity $A(.)$ is (implicitly) applied. The AWGN vector at the $k^{th}$ instant is denoted by $\textbf{n}_{k}$. At the receiver, for each user $u$, the precoded transmission affected by nonlinearity is received by the $u^{th}$ photodiode array. Then, according to the QoS of the user,  an estimate of the superposition of users' signals $\hat{\textbf{l}}_{k}^{(u)}$ for each user $u$ is recovered by multiplying with $(\hat{\textbf{H}}^{(u)}\hat{\textbf{P}}^{(u)})^{\dagger}$ to form an estimate of the superposition coded input $\textbf{x}_{k}$ at each UE, given by $\hat{\textbf{l}}_{k}^{(u)}$. After this, the SIC is performed at each UE to recover the respective user's symbols. In simulations, we consider the errors induced at each SIC layer. The $\hat{l}_{k}$, which is chosen among the $\hat{\textbf{l}}_{k}^{(u)}$ from UEs with the highest QoS and the input signal is used in the feedback loop (which can be an RF uplink \cite{burchardt2014vlc,o2008visible} or power line communication (PLC)-based \cite{ma2013integration,elgala2011indoor}) to adjust the pre-distorter coefficients $r_{k}^{(i)}$ using the NLMS or the Chebyshev-NLMS algorithms.
\section{existing adaptive pre-distortion techniques}\label{Sec3}
In this section, we review some recently proposed adaptive pre-distortion techniques which mitigates the nonlinear  characteristics of an LED in VLC. The discussion in this section would pave the way for the exposition on the proposed closed-loop nonlinear precoding technique, and its performance analysis in the subsequent sections.

Adaptive pre-distortion \cite{kim2014adaptive} is a popular technique for compensating the varying nonlinear characteristics of an LED in order to invert and track the varying nonlinear LED characteristics. Among adaptive pre-distortion techniques, the work in \cite{kim2014adaptive} assumes a scaling weight $r_{k}$ as a pre-distorter, which is multiplied with the input symbol sequence $x_{k}$ at the transmitter. Consequently,
it passes through the LED nonlinearity $A(.)$ and measurement noise $n_{k}$ is added so as to form an estimated signal $\hat{l}_{k}$.
Consequently, the $r_{k}$ is updated as the following NLMS-based algorithm as proposed in \cite{kim2014adaptive}:
\begin{equation}
r_{k+1} = r_{k} + \eta_{k}e_{k}x_{k}
\end{equation}
where, $\eta_{k}=\frac{\eta}{|x_{k}|^2}$ is the step-size for the NLMS algorithm \cite{sayed2003fundamentals} at the $k^{th}$ iteration given in
\cite{sayed2003fundamentals,farhang2013adaptive} ($\eta$ being a small positive number), $e_{k} = \beta x_{k} - \hat{l}_{k}$ ($\beta$ being a biasing constant), where $\hat{l}_{k}$ is the output given by
$\hat{l}_{k}=A(r_{k}x_{k})+n_{k}$.
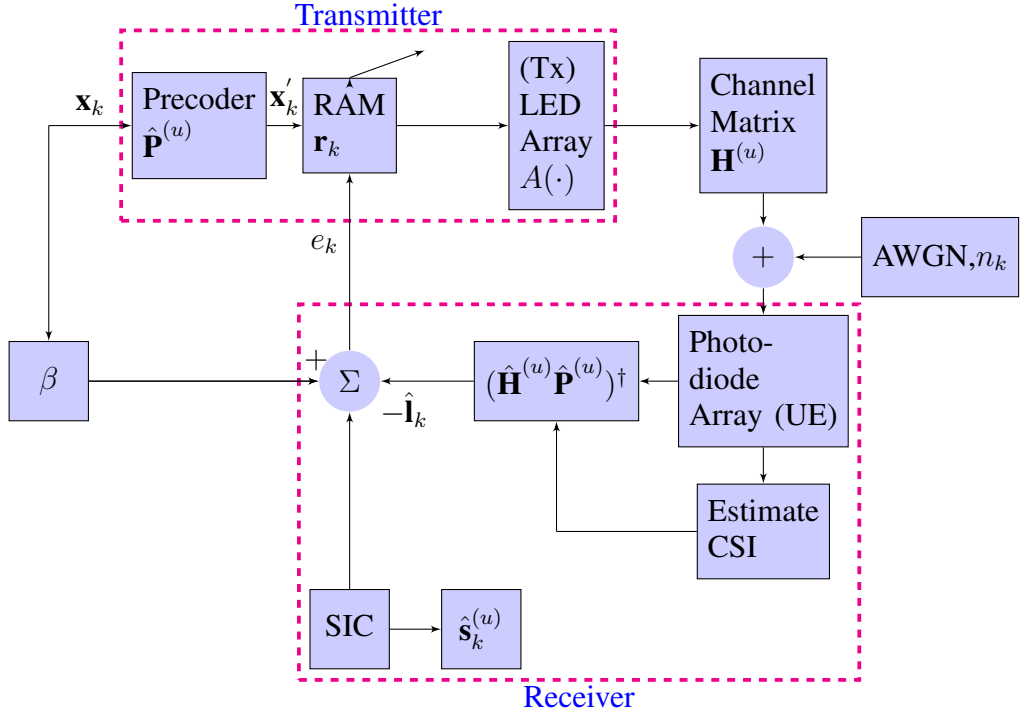
\begin{figure*}
\centering
\begin{tikzpicture}[auto, node distance=2cm,>=latex']
    \node [input, name=input] {};
    \node [block, right of=input] (PREC) {\symbolB \makecell[l]{Precoder \\$\hat{\textbf{P}}^{(u)}$}};
    \node [block, right of=PREC] (RAM) {\symbolB \makecell[l]{RAM\\$\textbf{r}_{k}$}};
    \node [block, right of=RAM,node distance=2.75cm] (LED) {\symbolB \makecell[l]{(Tx)\\LED\\ Array\\$A(\cdot)$}};
    \node [block, right of=LED,node distance=2.75cm] (H) {\symbolB \makecell[l]{Channel \\ Matrix\\$\textbf{H}^{(u)}$}};
    \node [block, below of=H,node distance=3.4cm] (PD) {\makecell[l]{\symbolB Photo-\\diode \\Array (UE)}};
    \node [circle, below of=H,node distance=1.75cm,fill=blue!20] (noise) {$+$};
    \node [block, right of=noise,node distance=2.35cm] (AWGN) {AWGN,$n_{k}$};
    \node [block, below of=LED,node distance=3.4cm] (PINVH) { \makecell[l]{ {$(\hat{\textbf{H}}^{(u)}\hat{\textbf{P}}^{(u)})^{{\dagger}}$}}};
    \node [block, below of=input,node distance=3.4cm] (scale) {$\beta$};
    \node [circle, below of=RAM,node distance=3.4cm,fill=blue!20] (summer) {$\Sigma$};
    \node [block, below of=summer,node distance=3.3cm] (SIC) {\makecell[l]{SIC}};
    \node [block, right of=SIC,node distance=1.75cm] (sym) {\makecell[l]{$\hat{\textbf{s}}_{k}^{(u)}$}};
    \node [block, below of=PD,node distance=2cm] (PQ) {\makecell[l]{Estimate \\CSI}};
    \node[line width=0.5mm,dashed,draw=magenta, fit= (PREC) (RAM)  (LED),label={[blue,label distance=-0.1cm]above:Transmitter}] (BB1){};
    \node[line width=0.5mm,dashed,draw=magenta, fit= (PINVH) (PD) (sym) (SIC),label={[blue,label distance=-0.1cm]south:Receiver}] (BB2){};

    \draw [draw,->] (input) -- node {$\textbf{x}_{k}$} (PREC);
    \draw [draw,->] (PREC) -- node {$\textbf{x}_{k}^{'}$} (RAM);
    \draw [draw,->] (RAM) -- node {} (LED);
    \draw [draw,->] (LED) -- node {} (H);
    \draw [draw,->] (H) -- node {} (noise);
    \draw [draw,->] (noise) -- node {} (PD);
    \draw [draw,->] (PD) -- node {} (PINVH);
    \draw [draw,->] (PINVH) -- node [near end] {$-\hat{\textbf{l}}_{k}$} (summer);
    \draw [draw,->] (SIC) -- node {} (sym);
    \draw [draw,->] (PD) -- node {} (PQ);
    \draw [draw,->] (input) -- node {} (scale);
    \draw [draw,->] (scale) -- node {} (summer);
    \draw [draw,->] (summer) -- node[pos=0.6] {$e_{k}$}
        node [near end] {} node {} (RAM);
        \draw [->] (scale) -- node[pos=0.99] {$+$}
        node [near end] {} (summer);
    \draw [->] (SIC) -- node[pos=0.49] {}
         (summer);
     \draw[->]
            (RAM.north) edge (5,1);
      \draw[-]           (PQ.west) -- (6.75,-5.4)  node[left] {};
      \draw [->]
             (6.75,-5.4)-- (PINVH.south);
                   \path[draw,->]
            (AWGN) edge (noise);
\end{tikzpicture}
\caption{Block diagram of the proposed system model.}\label{Fig_2}
\end{figure*}
However, the scaling factor based NLMS algorithm does not incorporate higher order statistics, which is essential to mitigate nonlinear transfer characteristics of the LED. Hence, NLMS based Chebyshev regression has been suggested recently \cite{mitra_chebyshev} where $\hat{\textbf{l}}_{k}=A(\sum_{\forall i}r_{k}^{(i)}T_{i}(\textbf{x}_{k}))+\textbf{H}^{(u)^{\dagger}}\textbf{n}_{k}$ (where $T_{i}(.)$ is the $i^{th}$ Chebyshev polynomial), and the squared error, $e_{k}^2 = (\beta x_{k} - \hat{l}_{k})^2$, is minimized. The reason for the choice of Chebyshev expansion in \cite{mitra_chebyshev} is twofold: a) Chebyshev expansion is a nonlinear function of the input and is therefore more suited to invert nonlinear characteristics of LED, and b) the Chebyshev expansion optimizes the min-max approximation error optimality criterion which is one of the most important reasons for preferring Chebyshev expansion based pre-distortion over other orthogonal polynomial-expansion based pre-distortion techniques. This makes the Chebyshev pre-distorter a better and more suitable pre-distortion technique (as compared to lookup table and adaptive NLMS based pre-distortion as suggested in \cite{kim2014adaptive}).
\section{Existing two-users' QoS guaranteed QR-precoded NOMA}\label{Sec4}
In this section, a technique for two-user NOMA is reviewed, where both users have similar channels \cite{ding}. Such a scenario is characterized by diverse QoS requirements, not necessarily by diverse channel conditions. In the scenario considered in this paper, we generalize NOMA to arbitrary number of users in a correlated-channel scenario, and hence it will be insightful to review the work given in \cite{ding} prior to the proposed algorithm in Section-\ref{Sec5}.
From \cite{ding}, a two user scenario is considered in which both users experience possibly similar channels,
$\textbf{H}^{(u)}$, with $u\in \{1,2\}$ but with varying QoS requirements.
\begin{equation}
\textbf{x}_{k} = \sum_{u=1}^{2}\sqrt{P^{(u)}}\textbf{s}^{(u)}_{k}
\end{equation}
In \cite{ding} a user (say User 2) is selected, whose experience we would like to improve selectively, and the QR decomposition of the transpose of the user's corresponding channel is considered such that $\textbf{H}^{(2)^{T}}=\textbf{Q}^{(2)}\textbf{R}^{(2)}$. Here, $\textbf{Q}^{(2)}\in \mathbb{R}^{M_{R}\times M_{R}}$ is a unitary matrix and $\textbf{R}^{(2)}\in\mathbb{R}^{M_{R}\times M_{T}}$ is an upper triangular matrix obtained by QR-decomposition. 
Next, the superposition coded vector is precoded by the matrix $\textbf{P}^{(u)}$ prior to transmission (in other words $\textbf{x}_{k}^{'}=\textbf{Q}^{(2)}\textbf{x}_{k}$). As a result, User 2 experiences a lower triangular matrix which facilitates recovery by SIC at its corresponding UE. Thus, User 2, which has a higher QoS requirement, experiences a diversity gain.

On the other hand, User 1's signal is recovered by the zero forcing solution upon $\textbf{x}_{k}^{'}$ by the matrix $(\textbf{H}^{(1)}\textbf{Q}^{(2)})^{\dagger}$. This step results in inducing different channel condition at UE with lesser QoS due to multiplication by the Wishart matrix corresponding to $(\textbf{H}^{(1)}\textbf{Q}^{(2)})^{\dagger}$. In other words, the multiplication of $(\textbf{H}^{(1)}\textbf{Q}^{(2)})^{\dagger}$, i.e. the equivalent channel condition, at User 1 is degraded as compared to that of User 2, to the extent controlled by the individual users' QoS. This is done such that User 1's QoS are still met while improving the performance of User 2 through QR-precoding and QoS based power allocation \cite{ding}.
\section{Proposed Chebyshev pre-distortion for MIMO NOMA-VLC systems}\label{Sec5}
In this work, a MIMO-VLC system is considered, in which it is assumed that all UEs are equipped with a photodiode array. In order to learn the pre-distorter weights in a MIMO system, we need to minimize the cost function, $J_{\text{MIMO}}=\min \limits_{r_{k}^{(i)}}\mathbb{E}[(\|\beta \textbf{x}_{k}-\hat{\textbf{l}}_{k}\|_{2}^2)]$, involving Euclidean norm, with respect to the weights $\textbf{r}_{k}=[r_{k}^{(i)}]$.
$J_{\text{MIMO}}$ could also be written as:
\begin{align}
J_{\text{MIMO}} &= \sum_{\forall x_{k}\in\textbf{x}_{k},\hat{l}_{k}\in\hat{\textbf{l}}_{k}}(\beta x_{k}-\hat{l}_{k})^2\\ \nonumber
\hspace{0.5cm}\text{s.t.} \hspace{1cm}& \Re[\sum_{\forall i}r_{k}^{(i)}T_{i}(\textbf{x}_{k})] > 0, \\ \nonumber
&\Im[\sum_{\forall i}r_{k}^{(i)}T_{i}(\textbf{x}_{k})] > 0
\end{align}
As the Euclidean norm is convex in the Chebyshev coefficient weights, stochastic gradient descent based adaptation is guaranteed to converge to the global optimum of the cost function. The above constraints link this technique to the average power constraint as given in \cite{zhang2016bandlimited}. For the peak power constraint mentioned in \cite{zhang2016bandlimited}, we assume soft clipping by the Rapp LED-nonlinearity model which has been used in the literature to model white LEDs \cite{elgala2010led,mitra2016adaptive} (OSRAM, Golden DRAGON, LA W57B, LY W57B), and mitigate this distortion by adaptive pre-processing.

From Fig. \ref{Fig_2}, the estimate $\hat{l}_{k}$ is derived from a feedback uplink from the receiver (which can be achieved by time-division duplexing, frequency division duplexing or by a separate RF/PLC uplink \cite{burchardt2014vlc,o2008visible}). Each UE sends an estimate, $\hat{l}_{k}^{(u)}$, to the transmitter.  The transmitter chooses $\hat{l}_{k}$ estimates from the UE with the best QoS since it has the maximum signal to interference and noise ratio (SINR).
The pre-distorter coefficients $r_{k}^{(i)}$ are initialized by a feasible weight vector found by random search, and then updated by taking gradient of $J_{\text{MIMO}}$ with respect to $r_{k}^{(i)}$ via a similar stochastic gradient NLMS based approach as:
\begin{gather}
 p_{k+1}^{(i)} = r_{k}^{(i)} + \sum_{\forall x_{k}\in\textbf{x}_{k},\hat{l}_{k}\in\hat{\textbf{l}}_{k}}\frac{\eta}{\sum_{\forall i}T_{i}(x_{k})^2}e_{k}T_{i}(x_{k})
 \end{gather}
$\forall x_{k}$ in $\textbf{x}_{k}$.
Consequently, the positivity of the Chebyshev expansion is enforced by projecting the gradient to the second quadrant in the Argand plane as follows:
\begin{align}  
\textbf{r}_{k+1} & = \textbf{p}_{k+1}, \hspace{0.25cm}\text{if $\Re[\sum_{\forall i}p_{k+1}^{(i)}T_{i}(\textbf{x}_{k}) > 0$], $\Im[\sum_{\forall i}p_{k+1}^{(i)}T_{i}(\textbf{x}_{k}) > 0$]} \\ \nonumber
&=\textbf{r}_{k}, \hspace{0.3cm}\text{otherwise}
\end{align}
$r_{k}^{(i)}$ denotes the $i^{th}$ pre-distorter weight at $k^{th}$ time instant and $e_{k} = \beta x_{k}-\hat{l}_{k}$ for each $x_{k}$ in $\textbf{x}_{k}$.



In this paragraph, we prove the feasibility of a hybrid Chebyshev pre-distorter that is coupled with precoding  to extend its suitability to NOMA channels. The Chebyshev pre-distortion in \cite{mitra_chebyshev} is modified by adding a precoder $\textbf{P}^{(u)}$ at the transmitter for each user $u$ (in addition to the Chebyshev pre-distorter which mitigates the LED nonlinearity), and the symbols are recovered from $\hat{l}_{k}^{(u)}$ by multiplying with $(\textbf{H}^{(u)}\textbf{P}^{(u)})^{\dagger}$ and performing SIC at each user. Next, we show that the recovery of symbols is made possible by using this pre-distorter in the presence of LED nonlinearity. The output $\hat{\textbf{l}}_{k}^{(u)}$ can be written as follows:
\begin{gather}
\hat{\textbf{l}}_{k}^{(u)}=(\textbf{H}^{(u)}\textbf{P}^{(u)})^{\dagger}[\textbf{H}^{(u)}A(\sum_{\forall i}r_{k}^{(i)}T_{i}(\textbf{P}^{(u)}\textbf{x}_{k}))+\textbf{n}_{k}]
\end{gather}
From extension of Bussgang's theorem \cite{zhang2012general}, we can write 
\begin{gather}\label{buss}
\hat{\textbf{l}}_{k}^{(u)} = \alpha \textbf{x}_{k} + (\textbf{H}^{(u)}\textbf{P}^{(u)})^{\dagger}\textbf{H}^{(u)}\bm{\delta} + (\textbf{H}^{(u)}\textbf{P}^{(u)})^{\dagger} \textbf{n}_{k}
\end{gather}
where $\bm{\delta}$ is a zero mean uncorrelated distortion noise sequence with variance $\sigma_{\delta}^{2}$ and $\alpha$ is a scaling correlation factor between $A(T(\cdot))$ with its argument. $\textbf{H}^{(u){\dagger}}$ is the pseudo inverse of a matrix $\textbf{H}^{(u)}$ and is assumed to be a left-inverse of $\textbf{H}$. Hence, $\hat{l}_{k}^{(u)}$ lies in the same subspace of $x_{k}$. Hence any precoding as done in the case of the linear scenario in \cite{ding}, works for the nonlinear closed loop system model as well.

Next, we derive an expression for overall SNR for the proposed CR-inspired NOMA VLC system employing Chebyshev pre-distortion, which will be used in further sections to design suitable precoding matrices, and derive the power allocation strategy for the considered NOMA-MIMO scenario.

In order to study the impact of nonlinearity on the overall detection performance, let us define the following terms: $\textbf{Z}^{(u)} = (\textbf{H}^{(u)}\textbf{P}^{(u)})^{\dagger}$ and $\textbf{W}^{(u)} = (\textbf{H}^{(u)}\textbf{P}^{(u)})^{\dagger}\textbf{H}^{(u)}$. Thus, from (\ref{buss}) the $u^{th}$ user is decoded with the following SINR, $\Gamma^{(u)}$,
\begin{gather}
 \Gamma^{(u)} = \frac{\alpha^{2}P^{(u)}}{\Bigg(\splitfrac{\alpha^{2}\sum_{\forall b>u}P^{(b)}+\text{Tr}({\textbf{W}^{(u)^{T}}\textbf{W}^{(u)}})\sigma_{\delta}^{2}}{+\text{Tr}({\textbf{Z}^{(u)^{T}}\textbf{Z}^{(u)}})\sigma_{v}^{2}}\Bigg)}
\end{gather}
, where $\text{Tr}(.)$ denotes the trace operator. 

Let the noise floor be indicated by $\sigma_{v}^{{(u)'}^2} = \text{Tr}(\textbf{Z}^{(u)^{T}}\textbf{Z}^{(u)})\sigma_{v}^{2}$. The variance of $\delta$ denoted by $\sigma_{\delta}^{2}$ can be written as follows for the NLMS algorithm \cite{haykin2008adaptive}:
\begin{gather}
\sigma_{\delta}^{2} = \frac{\eta}{2}\sigma_{v}^{{(u)'}^2}
\end{gather}

The total error variance (excess mean squared error along with the noise floor) can be quantified as follows:
\begin{gather}  
\sigma_{o}^{{(u)^2}}=\frac{\text{Tr}({\textbf{W}^{(u)^{T}}\textbf{W}^{(u)}})\sigma_{\delta}^{2}+\sigma_{v}^{{(u)'}^2}}{\alpha^2}
\end{gather}
Upon convergence of Chebyshev coefficients to optimal values (as $\alpha \to 1$), the SINR for the $u^{th}$ user is as follows:
\begin{gather}\label{sinr}
 \Gamma^{(u)} = \frac{P^{(u)}}{\sum_{\forall b>u}P^{(b)}+\sigma_{o}^{(u)^2}}
\end{gather}
Please note that in the above expression, the transmitted symbols of each user is assumed to be normalized to unit power, and $b$ is the index variable denoting the SIC layer.
\section{proposed precoder design}\label{Sec6}
In this section, design of the set of precoding matrices is described in detail for the considered correlated MIMO NOMA-VLC channel scenario. In (\ref{sinr}), the SINR experienced by the $u^{th}$ user depends on the pseudo-inverse of the overall channel matrix $\textbf{H}^{(u)}\textbf{P}^{(u)}$. 
Traditional beamforming solutions (as in \cite{ding2015application}) lack suitability, as some of the channel matrices have correlated eigenvectors, thereby reducing the available degrees of freedom to accommodate more users. 

However, in MIMO NOMA-VLC, the ill-conditioned channel matrix is used to allow us to have many levels of QoS for each user by the precoding technique described below. To design a precoding matrix $\textbf{P}^{(u)}$, let us consider SVD of the $u^{th}$ users' channel matrix $\textbf{H}^{(u)}$ as follows:
\begin{gather}
\textbf{H}^{(u)}=\textbf{U}^{(u)}\bm{\Sigma}^{(u)}\textbf{V}^{(u)^{T}}
\end{gather}
where
\begin{gather}
\bm{\Sigma}^{(u)}=\text{diag}(\sigma^{(1,(u))},\sigma^{(2,(u))},...,\sigma^{(c^{(u)},(u))},0,0,0...)
\end{gather}
$c^{(u)}$ denotes the rank of the $u^{th}$ user's matrix. The matrices $\textbf{U}^{(u)}$ and $\textbf{V}^{(u)}$ consists of eigenvectors of $\textbf{H}^{(u)}\textbf{H}^{(u)^{T}}$ and  $\textbf{H}^{(u)^{T}}\textbf{H}^{(u)}$. 
Thus the precoding matrix $\textbf{P}^{(u)}$ for each user would be given by:
\begin{gather}
\textbf{P}^{(u)}=\textbf{V}^{(u)}\bm{\Sigma}^{(u)^{\lambda^{(u)}-1}}
\end{gather}
where $\lambda^{(u)}$ is the exponent assigned to $u^{th}$ user and
\begin{gather}
\bm{\Sigma}^{(u)^{\lambda^{(u)}-1}}= \nonumber\text{diag}(\sigma^{{(1,(u))}^{\lambda^{(u)}-1}},\sigma^{{(2,(u))}^{\lambda^{(u)}-1}},...,\\ \nonumber
\sigma^{({c^{(u)},(u))}^{\lambda^{(u)}-1}},0,0,0...)
\end{gather}
Upon such precoding, each user would experience a virtual ``parallel" channel given by the following equation (assuming $\alpha\to1$ and $\sigma_{\delta}^2\to 0$):
\begin{gather}\label{eqsys}
\hat{\textbf{l}}_{k}^{(u)}\approx\alpha \textbf{x}_{k}+(\textbf{U}^{(u)}\bm{\Sigma}^{(u)^{\lambda^{(u)}}})^{\dagger}\textbf{n}_{k}
\end{gather}
Thus each user would experience a channel $\textbf{U}^{(u)}\bm{\Sigma}^{(u)^{\lambda^{(u)}}}$ with a distinctive trace and hence inducing diverse channel conditions with proper choice of $\lambda^{(u)}$. Further, as in \cite{ding}, we may want to improve the condition number of one of the user's (say $u_{1}$) channel with a higher QoS by QR factorization technique as done in \cite{ding}, i.e. we may assign $\textbf{P}^{(u)}
=\textbf{V}^{(u)}\bm{\Sigma}^{(u)^{\lambda^{(u)}-1}}\textbf{Q}^{(u_{1})}$, where
$\textbf{U}^{(u_{1})}\bm{\Sigma}^{(u_{1})^{\lambda^{(u_{1})}}}=\textbf{R}^{(u_{1})^T}\textbf{Q}^{(u_{1})^{T}}$. However, the QR factorization technique is not essential as the generalized power diversity is achieved by varying $\lambda^{(u)}$. The power allocation strategy for the $u^{th}$ user according to varying levels of user-QoS is detailed in the next section.
\section{proposed power allocation strategy}\label{Sec7}
In this section, we derive the power allocation strategy that needs to fulfill a given QoS for a user. QoS for the $u^{th}$ user is typically designated by a reliable transmission rate $R^{(u)}$. This can be written as:
\begin{gather}
 \log_{2}(1+\Gamma^{(u)})\geq R^{(u)}
\end{gather}
Let us define $\epsilon^{(u)}=2^{R^{(u)}}-1$. Thus,
\begin{gather}
\frac{P^{(u)}}{\sum_{\forall b>u}P^{(b)}+\sigma_{o}^{(u)^2}}\geq \epsilon^{(u)}
\end{gather}
Rearranging terms we get,
\begin{gather}
P^{(u)}\geq \epsilon^{(u)} \sum_{b>u}P^{(b)} + \epsilon^{(u)}\sigma_{o}^{(u)^2}
\end{gather}
adding $\sum_{b>u}P^{(b)}$ on both sides:
\begin{gather}
1\geq P^{(u)}+\sum_{b>u}P^{(b)}\geq (1+\epsilon^{(u)})\sum_{b>u}P^{(b)}+
\epsilon^{(u)}\sigma_{o}^{(u)^2}
\end{gather}
The above expression for the residual interference at $b^{th}$ layer of SIC can be re-written as follows:
\begin{gather}
\sum_{b>u}P^{(b)}\leq \frac{1-\epsilon^{(u)}\sigma_{o}^{(u)^2}}{1+\epsilon^{(u)}}
\end{gather}
Assuming $1-P^{(u)}\geq 0 $, the following power allocation for each user is:
\begin{gather}\label{powalloc}
P^{(u)} = \min \Bigg(1,\epsilon^{(u)}\frac{(1+\sigma_{o}^{(u)^2})}{1+\epsilon^{(u)}}\Bigg)
\end{gather}
It is to be noted that this power allocation strategy is derived by considering users' differing QoS requirements with each user having a specific $\textbf{Z}^{(u)}$ depending on its requirement $\epsilon^{(u)}$, and is the minimum power needed to meet each user's QoS. Thus this power-allocation technique is a better and generalized power allocation technique as compared to the gain ratio power allocation in \cite{hanaa} for the considered MIMO NOMA-VLC channel considering correlated channels. Further, since the lower bound is chosen to allocate the power in (\ref{powalloc}), the proposed power allocation assigns the minimum power that satisfies the individual constraints of users, and causes minimal interference to users in previous layer in SIC. 
\subsection{Choice of $\lambda^{(u)}$ for each user}
In the considered NOMA system, it is very important to induce power diversity, which in this work, is provided by precoding. We claim the following update rule for assigning $\lambda^{(u)}$ to the $u^{th}$ user:
\begin{gather}\label{rule}
{\lambda^{(u+1)}} = {\lambda^{(u)}\frac{\log(\sum_{g=1}^{c^{(u)}}|\sigma^{(g,(u))}|^2)}{\log(\sum_{g=1}^{c^{(u+1)}}|\sigma^{(g,(u+1))}|^2)}}
\end{gather}
Next, an insight into the assignment of $\lambda^{(u)}$ from the perspective of statistical mechanics based learning, is provided by the following theorem. 
\begin{thm}\label{app1}
\it{Choice of $\lambda$ according to (\ref{rule}) maximizes the increment in sum-rate by admitting another user in the proposed  MIMO CR-NOMA-VLC system.}
\end{thm}
\begin{proof}
Please refer to the Appendix-A.
\end{proof}
\section{Analytical expression for BER of square $M$-QAM}\label{Sec8}
In this section, an analytical upper bound for BER of $M$-QAM is derived for the proposed precoding algorithm. This analytical expression is necessary to predict the system performance without computationally intensive Monte-Carlo simulations and for overall calibration of the wireless link for link-optimization \cite{dixit2014performance}. We consider square $M$-QAM, because of its widespread use in the VLC literature \cite{lee2012visible,tsonev20143,chow2013adaptive}.

Let us denote a square $M$-QAM constellation with amplitude and phase set given by $\{A_{n}\}_{n=1}^{\sqrt{M}},\{\phi_{m}\}_{m=1}^{\sqrt{M}}$ such that $M$-QAM constellation is expressed as $\{A_{n}\exp(j\phi_{m})\}_{n,m=1}^{\sqrt{M}}$. Hence, at each layer, we can have the following recursion for the respective set of modulii $\mu_{b,m,n}$ for $\hat{l}_{k}$ (using laws of vector addition):
\begin{gather}
\mu_{1,(m,n)} = \\ \nonumber
\sqrt{P^{(1)}|A_{m}|^2+P^{(2)}|A_{n}|^2+ 2\sqrt{P^{(1)}P^{(2)}}A_{m}^{*}A_{n}\cos(\phi_{n})}\\ \nonumber
\mu_{b+1,(m,n)} = \{\{A_{n}\exp(j\phi_{m})\} \boxplus \mu_{b,(m,n)}\}_{n,m=1}^{\sqrt{M}}\\ \nonumber \sqrt{P^{(b)}|A_{m}|^2+|\mu_{b,(m,n)}|^2+ 2\sqrt{P^{(b)}}\mu_{b,(m,n)}A_{m}^{*}\cos(\phi_{n})}\\ \nonumber
\forall m,n=1,2,\cdots, \sqrt{M}
\end{gather}
where we denote vector addition operation by $\boxplus$.
Finally the probability of bit-error
$P_{\sqrt{M}}$ for an equivalent $\sqrt{M}$-PAM modulation scheme could be written as (i.e. the BER of square $M$-QAM modulation can be assumed as two independent $\sqrt{M}$-PAM modulation \cite{proakis2007fundamentals}):
\begin{gather}\label{expression}
P_{\sqrt{M}}\approx 2\Big(\frac{\sqrt{M}-1}{\sqrt{M}}\Big)\sum_{\forall b}\sum_{\forall m}\sum_{\forall n} Q\Bigg(\sqrt{\frac{|\mu_{b+1,(m,n)}-\mu_{b,(m,n)}|^2}{\sigma_{o}^{(u)^2}}}\Bigg)
\end{gather}
The approximation has been made under the assumption of Gaussian residual $\bm{\delta}$. Typically, $\bm{\delta}$ can assume any distribution, and hence its BER characteristics would be difficult to obtain. However, since the Gaussian distribution maximizes the differential entropy given a covariance matrix, the Gaussian assumption gives us an upper bound for the BER expression. 
Finally, the expression for square $M$-QAM can be given by the following equation by considering the real and imaginary parts of the $M$-QAM constellation as independent:
\begin{gather}\label{qam}
P_{M}\approx1-(1-P_{\sqrt{M}})^{2}
\end{gather}
Thus (\ref{qam}) gives us an expression for probability of error for square $M$-QAM for the proposed SVD based precoding technique in a NOMA-VLC scenario.
For the special case of 4-QAM, the BER expression can be written as:
\begin{gather}\label{qam4}
P_{\text{QAM}_{\sqrt{M}=2}}\approx\sum_{\forall b}\sum_{m=1}^{\sqrt{M}}\sum_{n=1}^{\sqrt{M}}Q\Bigg(\sqrt{\frac{|\mu_{b+1,(m,n)}-\mu_{b,(m,n)}|^2}{\sigma_{o}^{(u)^2}}}\Bigg)
\end{gather}
and,
\begin{gather}\label{qam44}
P_{\text{QAM}_{\sqrt{M}=2}}\approx1-(1-P_{\text{QAM}_{\sqrt{M}=2}})^{2}\approx 2 P_{\text{QAM}_{\sqrt{M}=2}}
\end{gather}
where the relation between $\mu_{b+1}$ and $\mu_{b}$ can be written as follows (considering the amplitude set to be $\pm 1$ and phase set to be $\pm\frac{\pi}{4}$) :
\begin{gather}
\mu_{b+1}=\sqrt{P^{(b)}+|\mu_{b}|^2+\sqrt{2P^{(b)}}\mu_{b}}
\end{gather}
It can be readily seen that the above equation is the special case of (\ref{expression})
assuming the amplitude set to be only modulii $\{\pm 1\}$ and phase $\pm\frac{\pi}{4}$.
\subsection{BER in presence of estimation error}
In this section, we derive the BER for the proposed MIMO NOMA-VLC system for $M$-QAM in the presence of estimation error. We assume imperfect CSI both at the transmitter and the receiver. Let, $\bm{\Delta\Sigma_{1}}$ and $\bm{\Delta\Sigma_{2}}$ denote the non-diagonal error matrices at transmitter and the receiver respectively, such that $\hat{\textbf{H}}^{(u)}\hat{\textbf{V}}^{(u)}=\textbf{U}^{(u)}(\bm{\Sigma}^{(u)}+\bm{\Delta\Sigma_{1}})$, and $\hat{\textbf{U}}^{(u)^{T}}\hat{\textbf{H}}^{(u)} = (\bm{\Sigma}^{(u)}+\bm{\Delta\Sigma_{2}})\textbf{V}^{(u)}$. At the steady state, we transmit the following precoded broadcast:
\begin{gather}
\textbf{x}_{k}^{'} = \textbf{V}^{(u)}(\bm{\Sigma}^{(u)}+\bm{\Delta\Sigma_{1}})^{\lambda^{(u)}-1}\textbf{x}_{k}
\end{gather}
This can be approximated as follows under the assumption of small $\bm{\Delta\Sigma_{1}}$:
\begin{gather}
\textbf{x}_{k}^{'} \approx \textbf{V}^{(u)}\bm{\Sigma}^{(u)^{\lambda^{(u)}-1}}(\textbf{I}+(\lambda^{(u)}-1)\bm{\Delta\Sigma_{1}}^{-1}\bm{\Sigma_{1}})\textbf{x}_{k}
\end{gather}
Upon passing $\textbf{x}_{k}^{'}$ through the proposed pre-distorter, the LED nonlinearity, and the channel $\textbf{H}^{(u)}$, the received vector can be written as follows:
\begin{gather}
\hat{\textbf{l}}_{k}\approx (\textbf{I}-\lambda^{(u)}\bm{\Sigma}_{2}^{-1}\bm{\Delta\Sigma}_{2})(\textbf{I}+(\lambda^{(u)}-1)\bm{\Delta\Sigma_{1}}^{-1}\bm{\Sigma_{1}})\textbf{x}_{k}+
(\textbf{U}^{(u)}\bm{\Sigma}^{(u)^{\lambda^{(u)}}})^{\dagger}\textbf{n}+\textbf{W}^{(u)}\bm{\delta}
\end{gather}
Considering only the first order terms, the effective noise power can we rewritten as:
\begin{gather}
\sigma_{o}^{{'}^2}=\sigma_{o}^2 + \lambda^{(u)^2}\mathbb{E}[\text{Tr} [\textbf{x}_{k}\textbf{x}_{k}^{T}\bm{\Delta\Sigma}_{2}^{T}\bm{\Sigma}_{2}^{-T}\bm{\Sigma}_{2}^{-1}\bm{\Delta\Sigma}_{2}]] + (1-\lambda^{(u)})^{2}\mathbb{E}[\text{Tr} [\textbf{x}_{k}\textbf{x}_{k}^{T}\bm{\Delta\Sigma}_{1}^{T}\bm{\Sigma}_{1}^{-T}\bm{\Sigma}_{1}^{-1}\bm{\Delta\Sigma}_{1}]] 
\end{gather}
Assuming independent statistics of $\|\textbf{x}_{k}\|_{2}^{2}$, $\bm{\Delta\Sigma}_{2}$, and $\bm{\Delta\Sigma}_{1}$,
we can write as follows:
\begin{gather}
\sigma_{o}^{{(u)'}^2}=\sigma_{o}^{(u)^2} + \sigma_{\gamma}^{2}\lambda^{(u)^2}\text{Tr} [\bm{\Sigma}_{2}^{-T}\bm{\Sigma}_{2}^{-1}]\mathbb{E}[\|\textbf{x}_{k}\|_{2}^{2}] + \sigma_{\gamma}^{2}(1-\lambda^{(u)})^{2}\text{Tr} [\bm{\Sigma}_{1}^{-T}\bm{\Sigma}_{1}^{-1}]\mathbb{E}[\|\textbf{x}_{k}\|_{2}^{2}] 
\end{gather}

Thus the approximate expression for BER for $M$-QAM can be written as:
\begin{gather}\label{expression}
P_{\sqrt{M}}\approx2\Big(\frac{\sqrt{M}-1}{\sqrt{M}}\Big)\sum_{\forall b}\sum_{\forall m}\sum_{\forall n} Q\Bigg(\sqrt{\frac{|\mu_{b+1,(m,n)}-\mu_{b,(m,n)}|^2}{\sigma_{o}^{(u)'^2}}}\Bigg)
\end{gather}
Finally, the expression for square $M$-QAM\footnote[9]{For imperfect SIC from previous layers, the expression for BER can be rewritten by changing/expanding the definition of $\mu_{b}$. If a residual interference $I^{(b)}$ exists at layer $b$, $\forall b<u$ then $\mu_{b+1,(m,n)} = \{\{A_{n}\exp(j\phi_{m})\} \boxplus \mu_{b,(m,n)}\}_{n,m=1}^{\sqrt{M}} \boxplus \sum_{b<u} I^{(b)} P^{(b)}$}  can be given by the following equation by considering the real and imaginary parts of the $M$-QAM constellation independent:
\begin{gather}\label{qam}
P_{M}=1-(1-P_{\sqrt{M}})^{2}
\end{gather}

\section{Simulations}\label{Sec9}
A typical channel matrix was generated mathematically from \cite{hanaa}, for a room of size $5\text{m}\times5\text{m}\times 3\text{m}$, with refractive index of lens 1.5, height of LED-2.25m, area of the photodiode (PD) $1\text{cm}^{2}$, and the spacing between transmitter LED array photodiode array 0.4m. FOV of the photodetectors is kept fixed at 60 degrees. A $2\times2$ MIMO channel was considered for each user. All users were chosen close to each other such that all experience almost similar channel conditions. For \textit{power allocation} for each user, (\ref{powalloc}) was used in all simulations. These simulation parameters are tabulated in Table \ref{tabb}. For all the simulations $10^{6}$ symbols were considered with ensemble of 3000 Monte-Carlo runs using MATLAB.

The VLC quasi-static channel model \cite{ghassemlooy2012optical}, used for simulating the channel matrix is given by the following equation \cite{hanaa,wang2015multiuser} is used for simulations:
\begin{align}
[\textbf{H}]_{ij}^{(u)}=&\frac{\mathcal{A}_{e,i}}{d_{ij}^2\sin^{2}\Psi}R(\phi_{ij})\cos\theta_{ij}&&,0<\phi_{ij}<\Psi\\ \nonumber
=&0&&,\text{otherwise}
\end{align}
$\mathcal{A}_{e,i}$ denotes the area of the $i^{th}$ photodetector. $d_{ij}$ is the distance between $i^{th}$ transmitter LED and $j^{th}$ photodetector for each UE, $\phi_{ij}$ is the perpendicular angle of $j^{th}$ LED, $\theta_{ij}$ is the angle between $i^{th}$ transmit LED in transmitter-array and $j^{th}$ photodetector in a photodetector array with the receiver axis.
$\Psi$ denotes the FOV for each photodetector. $R(\phi_{ij})$ denotes the Lambertian radiant intensity which can be written as follows:
\begin{gather}
R(\phi_{ij})=\frac{(\kappa+1)\cos^{\kappa}(\phi_{ij})}{2\pi}
\end{gather}
$\kappa$ is the order of Lambertian emission given as follows:
\begin{gather}
\kappa = -\frac{\ln 2}{\ln(\cos(\phi_{\frac{1}{2}}))}
\end{gather}
Typically, as given in \cite{wang2015multiuser,fath2013performance}, these channel-matrices are inherently ill-conditioned and exhibit correlatedness when users are located close to each other.
\begin{center}
\begin{table}[!htbp]
\centering
\caption{Simulation Parameters}\label{tabb}
  \begin{tabular}{ | l |c| }
    \hline
    \scriptsize{Room-Size} & \scriptsize{$5 \text{m}\times 5 \text{m}\times 3 \text{m}$} \\ \hline
    \scriptsize{LED-Height} & \scriptsize{$2.25 \text{m}$} \\ \hline
    \scriptsize{Transmit LED location} & \scriptsize{(0.2,0,-0.75),(-0.2,0,-0.75)} \\ \hline
    \scriptsize{Photodetector-Area} & \scriptsize{$1 \text{cm}^2$} \\ \hline
    \scriptsize{LED array spacing} & \scriptsize{$0.4 \text{m}$}\\ \hline
    \scriptsize{LED Emission Half-Angle $\phi_{\frac{1}{2}}$} & \scriptsize{70 degrees}\\ \hline
    \scriptsize{$\eta$} & \scriptsize{$0.00022$} \\ \hline
    \scriptsize{User 1's PD array coordinates} &\scriptsize{$(0.1,0.1,-3), (0.1,-0.1,-3)$} \\ \hline
    \scriptsize{User 2's PD array coordinates} &\scriptsize{$(-0.1,0.1,-3), (-0.1,-0.1,-3)$} \\ \hline
    \scriptsize{User 3's PD array coordinates}  &\scriptsize{$(-0.35,0.35,-3), (-0.35,-0.35,-3)$} \\ \hline
    \scriptsize{User 4's PD array coordinates}  &\scriptsize{$(0.35,0.35,-3), (0.35,-0.35,-3)$} \\ \hline
    \scriptsize{FOV}, $\Psi$ & \scriptsize{60 degrees} \\ \hline
    \scriptsize{Ceiling Center} & \scriptsize{(0,0,0)} \\ \hline
  \end{tabular}
 \end{table}
\end{center}
Using the given simulation setup, we now describe the simulations that validates the proposed SVD-precoded Chebyshev-NLMS based pre-distorter. 
Using the proposed power allocation scheme derived in (\ref{powalloc}), the sum-rate for the proposed SVD-precoded Chebyshev-NLMS was observed by varying the number of users for 4-QAM and 16-QAM. Additionally, the proposed CR-NOMA based precoding/power allocation technique was also compared with GRPA based power allocation in terms of sum-rate (measured in bpcu (bits per channels use)). From Fig. \ref{fig_1} it is observed that the proposed technique yields a monotonous increase in the sum-rate as derived in the preceding analysis as the number of users is increased (as long as the constant power budget constraint is not violated). On adding a fourth user, the power required to meet the incoming user's QoS is violated, and hence only three users are shown in the plots. For GRPA, it can be observed that the sum rate curve reduces for $U>2$. This is because of two main reasons: a) GRPA does not consider the individual users' QoS into its power allocation technique, and b) GRPA does not perform well if the gain-ratios between individual channels are close to unity (i.e. if  the channels are similar). 

Next, the theoretical expression for average BER vs SNR derived in the previous sections for the proposed CR-NOMA based approach is validated via simulations and compared with GRPA, and the classical precoding technique given in \cite{ali2017non} and \cite{ding2016general}. Various levels of channel estimation error are considered for both 4-QAM and 16-QAM. The simulated BER is compared with theoretical expressions for BER, and also when a linear channel is considered. In Fig. \ref{fig_3}, the average BER for 4-QAM is compared against its corresponding analytically derived upper bound for estimation error of $\sigma_{\gamma}^{2}=0,33\text{dB},36\text{dB}$, thus demonstrating the accuracy of the theoretically derived expression. Further, from Fig. \ref{fig_3}, it can be observed that at high SNR, the average BER for the proposed approach comes close to the ideal linear channel scenario in perfect CSI scenario, which verifies that adaptive pre-distortion technique is compensating the LED nonlinearity at high SNR. Also, the proposed approach outperforms GRPA. In addition, similar simulations have been performed for the $U=3$ for 4-QAM (for estimation error $\sigma_{\gamma}^{2}=0,38\text{dB},40\text{dB}$) in Fig. \ref{fig_4}, and $U=2,3$ for 16-QAM ($\sigma_{\gamma}^{2}=50\text{dB},55\text{dB}$ for $U=2$, and $\sigma_{\gamma}^{2}=60\text{dB},65\text{dB}$ for $U=3$) in Fig. \ref{fig_5} and Fig. \ref{fig_6}. Especially for $U=3$, the BER expression for GRPA is not plotted as the sum-rate curve reduces for $U>2$ as observed in Fig. \ref{fig_1}. It can be further observed that at high SNR the simulated BER curves tend to the upper bound for BER derived in the preceding analysis for $M$-QAM, thus confirming their validity.

Also, it can be noted that at $U=3$ and 16-QAM modulation there is a gap of almost 1dB SNR between the BER performance at $24\text{dB}$ in the presence of nonlinearity with proposed pre-distortion, and the perfect linear channel. Despite using higher order modulation like 16-QAM, and increasing the number of users to $U=3$, it can be observed from Fig. \ref{fig_6} that the residual variance of the Chebyshev pre-distorter (which is directly proportional to noise power), has little effect at high SNR as the average BER performance of the proposed approach comes closer (within 1dB) to that of the ideal linear channel. Further,  small perturbations in the form of estimation errors (like 65dB) causes degradation in the overall average BER characteristics, as we increase the number of users, and use higher order modulation.

Finally, comparison with existing zero-forcing precoding as given in \cite{ali2017non} and \cite{ding2016general} is also presented for both 4-QAM and 16-QAM for $U=2,3$. Better BER performance is observed in all scenarios in case of the proposed SVD-based precoding as compared to the existing precoding techniques given in \cite{ali2017non} and \cite{ding2016general}.

\section{Conclusion}\label{Sec10}
In this work, the CR-inspired MIMO-NOMA technique for VLC has been theoretically investigated in the presence of LED nonlinearity. Initially, challenges for MIMO NOMA-VLC scenarios (like correlated channels, and device nonlinearity) have been reviewed. Next, a novel generalized precoding technique is proposed for IoT applications for arbitrary number of users with differing QoS requirements (but not necessarily different channel conditions). This precoding technique is combined with Chebyshev pre-distortion for applicability in LED nonlinearity impaired channels. A CR-inspired power-allocation algorithm is also proposed for the algorithm derived in this work. Analytical upper bounds on BER is carried out for the proposed precoding, and the derived CR-inspired power allocation technique for square $M$-QAM. Further, simulations have been performed to validate the theoretically derived analytical upper bounds. This work can have applications in VLC as an integral part of IoT and Li-fi devices where the applications of NOMA are more prevalent as we converge towards 5G technologies.

\appendix
\subsection{\textit{Proof of Theorem 1}}
It can be noted from (\ref{rule}) that $\lambda^{(u)}$ is chosen according to the equation:
\begin{gather}
{\lambda^{(u)}}\log(S^{(u)})=\kappa
\end{gather}
where $\kappa$ is an arbitrary positive constant, and $\log(\sum_{g=1}^{c^{(u)}}|\sigma^{(g,(u))}|^2)=S^{(u)}$, $S^{(u)}$ denoting the strength of the channel. Without loss of generality, let, 
\begin{align}
\log(S^{(u)})&=\frac{1}{\log(T^{(u)})}&,\log(S^{(u)})>0 \\ \nonumber
             &=-\frac{1}{\log(T^{(u)})}&,\log(S^{(u)})<0 \\ \nonumber
\end{align}
where $T^{(u)}$ is an auxiliary variable. Hence,
\begin{gather}
{\lambda^{(u)}} = \kappa \log(T^{(u)})
\end{gather}
For two users, user $b$, and user $b^{'}$, it can be written that:
\begin{gather}
{\lambda^{(b)}} = \kappa \log(T^{(b)})\\ \nonumber
{\lambda^{(b^{'})}} = \kappa \log(T^{(b^{'})})
\end{gather}
Then, the increment in $\lambda$, $\Delta \lambda = \lambda^{(b)} - \lambda^{(b^{'})}$ can be written as follows:
\begin{gather}
\Delta \lambda = \kappa \log\Bigg(\frac{T^{(b)}}{T^{(b^{'})}}\Bigg)
\end{gather}
This implies that,
\begin{gather}
T^{(b)} = T^{(b^{'})}\exp(\frac{\Delta\lambda}{\kappa})
\end{gather}
As the sum of powers of all the channels are finite, we can assume that $\sum_{\forall u}T^{(u)} = T$, where $T$ is a finite positive real number. Then the following can be noted:
\begin{gather}\label{gibbs}
\upsilon(\Delta\lambda)=\frac{T^{(b)}}{T} = \frac{\exp(-\frac{\Delta\lambda^{(b)})}{\kappa}}{\sum_{u}\exp(-\frac{\Delta\lambda^{(b)}}{\kappa})}
\end{gather}
where $\upsilon(\Delta\lambda)$ is a probability measure on $\Delta\lambda$.
As given in \cite{ding2016impact}, the rate-gap between two users, $b$ and $b^{'}$, for CR inspired NOMA, can be approximated at high SNR by the following equation:
\begin{gather}\label{gap}
R^{(b)}-R^{(b^{'})}\approx \log(|\bm{\Sigma}^{(b)^{2\lambda^{(b)}}}|)-\log(|\bm{\Sigma}^{(b^{'})^{2\lambda^{(b^{'})}}}|)\\ \nonumber
=\sum_{\forall g} \Bigg(\frac{\sigma^{(g,{(b)})}}{\sigma^{(g,{(b^{'})})}}\Bigg)^{2\lambda^{(b)}}+\sum_{\forall g}[\sigma^{(g,(b^{'}))}]^{2\Delta\lambda^{(b)}}
\end{gather}
From (\ref{gibbs}), we can write the probability density function (p.d.f) $\upsilon(\Delta\lambda)$ as follows:
\begin{gather}
\upsilon(\Delta\lambda)=\frac{1}{Z}\exp(-\frac{\Delta\lambda}{\kappa})
\end{gather}
where $Z$ is the partition function in the denominator which is a normalizing constant. It can be seen that the p.d.f follows a Gibbs distribution.
Then, the following can be inferred:
\begin{gather}
\log(\upsilon(\Delta\lambda)) = -\frac{\Delta\lambda}{\kappa} -\log Z
\end{gather}
Taking expectation on both sides with respect to $\upsilon(\Delta\lambda)$,
\begin{gather}\label{prove}
\kappa(H - \log Z) = \mathbb{E}[\Delta\lambda]
\end{gather}
where $H = -\mathbb{E}[\log(\upsilon(\Delta\lambda))]$ denotes the entropy with respect to $\upsilon(\Delta\lambda)$, which achieves maxima by sampling $\Delta\lambda$ with respect to the Gibbs distribution (which belongs to the class of maximum entropy distributions \cite{kindermann1980markov}) as in (\ref{gibbs}).

From (\ref{gap}), the gap in sum-rate with the expectation taken with respect to $\upsilon(\Delta\lambda)$ can be written as follows:
\begin{gather}
\mathbb{E}_{\upsilon(\Delta\lambda)}[R^{(b)}-R^{(b^{'})}] = \sum_{\forall g} \Bigg(\frac{\sigma^{{(g,(b))}}}{\sigma^{(g,{(b^{'})})}}\Bigg)^{2\lambda^{(b)}}+\mathbb{E}_{\upsilon(\Delta\lambda)}[\sum_{\forall g}[\sigma^{(g,(b^{'}))}]^{2\Delta\lambda^{(b)}}]
\end{gather} 
By Jensen's inequality:
\begin{gather}
\mathbb{E}_{\upsilon(\Delta\lambda)}[R^{(b)}-R^{(b^{'})}] \geq \sum_{\forall g} \Bigg(\frac{\sigma^{{(g,(b))}}}{\sigma^{{(g,(b^{'}))}}}\Bigg)^{2\lambda^{(b)}}+[\sum_{\forall g}[\sigma^{(g,(b^{'}))}]^{\mathbb{E}_{\upsilon(\Delta\lambda)}[2\Delta\lambda^{(b)}]}]
\end{gather}
Let $K = \sum_{\forall g} \Big(\frac{\sigma^{{(g,(b))}}}{\sigma^{{(g,(b^{'}))}}}\Big)^{2\lambda^{(b)}}$. From (\ref{prove}), it can be seen that:
\begin{gather}\label{final}
\mathbb{E}_{\upsilon(\Delta\lambda)}[R^{(b)}-R^{(b^{'})}]\geq K + [\sum_{\forall g}[\sigma^{(g,(b^{'}))}]^{2\kappa(H-\log Z)}]
\end{gather}
As the Gibbs distribution is a maximum entropy distribution which maximizes $H$, the rate gap between the users is maximized as can be inferred from (\ref{final}).

\ifCLASSOPTIONcaptionsoff
  \newpage
\fi
\bibliographystyle{IEEEtran}
\bibliography{IEEEabrv,paper}
\newpage
\begin{figure}[!htbp]
  \centering
  \includegraphics[width=1\linewidth,height=12cm]{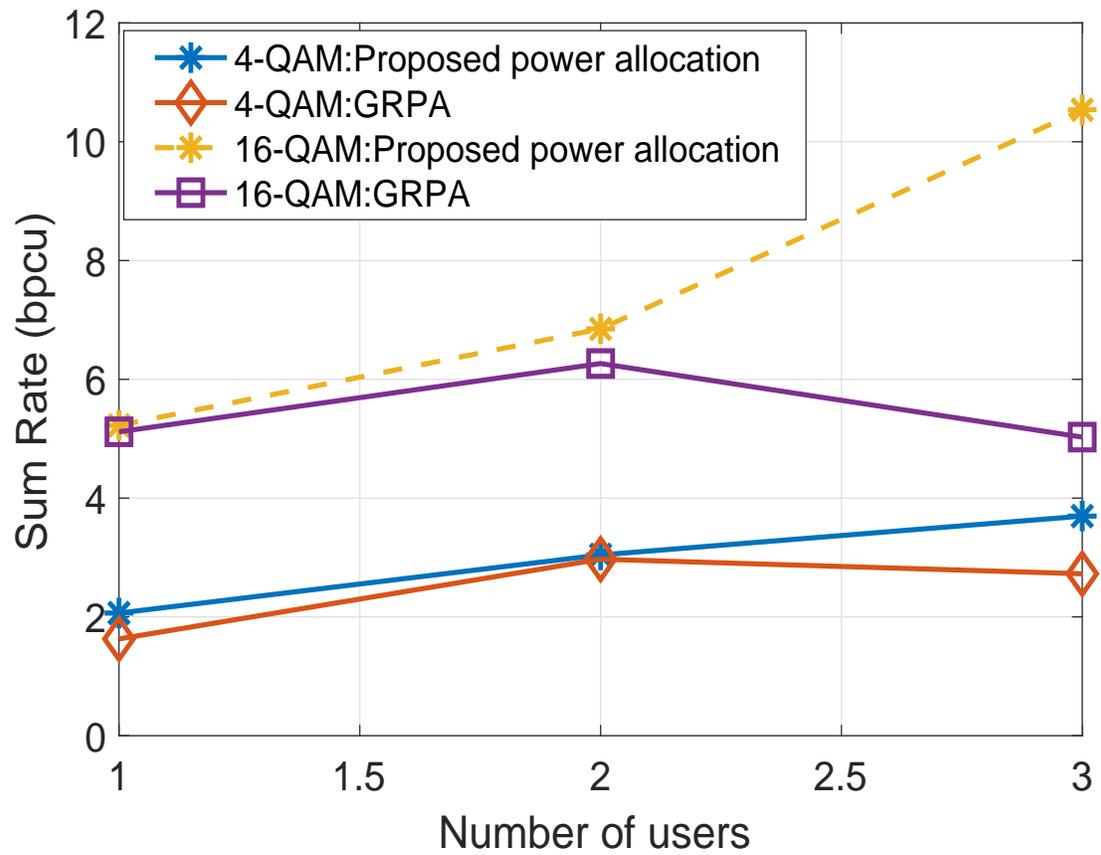}
  \caption{Sum-Rate vs number of users compared with GRPA for 4-QAM and 16-QAM.}\label{fig_1}
\end{figure}
\begin{figure}[!htbp]
  \centering
  \includegraphics[width=1\linewidth,height=10cm]{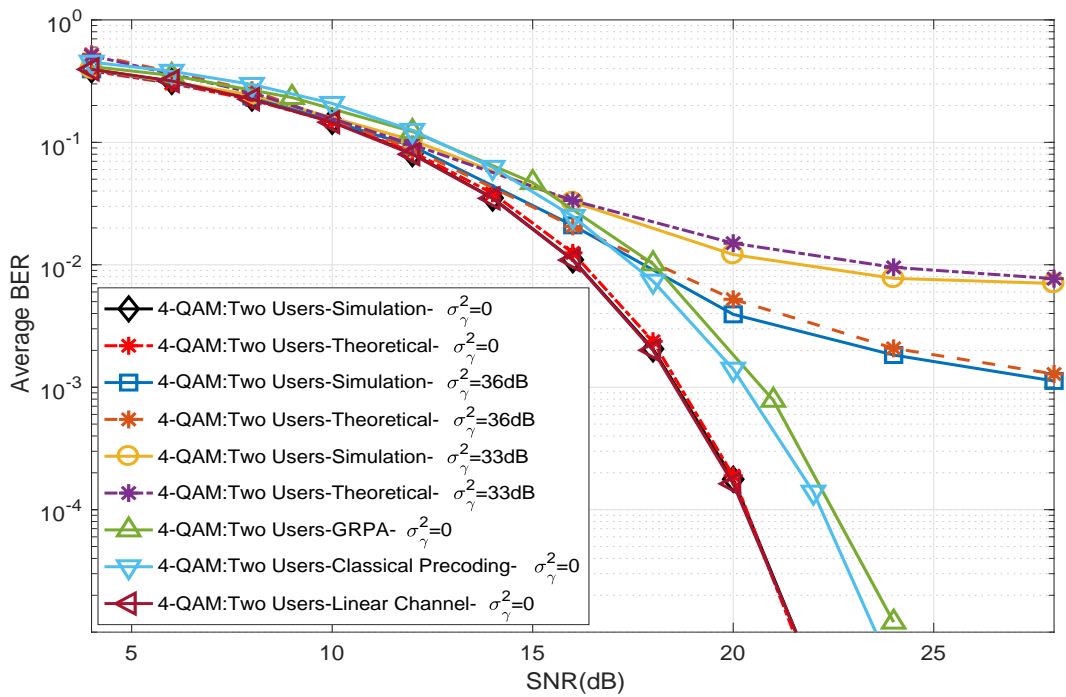}
  \caption{Plot of average BER vs SNR for the proposed approach, GRPA, and classical precoding  for $U=2$ for the proposed technique for various estimation error-levels for 4-QAM modulation.}\label{fig_3}
\end{figure}
\begin{figure}[!htbp]
  \centering
  \includegraphics[width=1\linewidth,height=12cm]{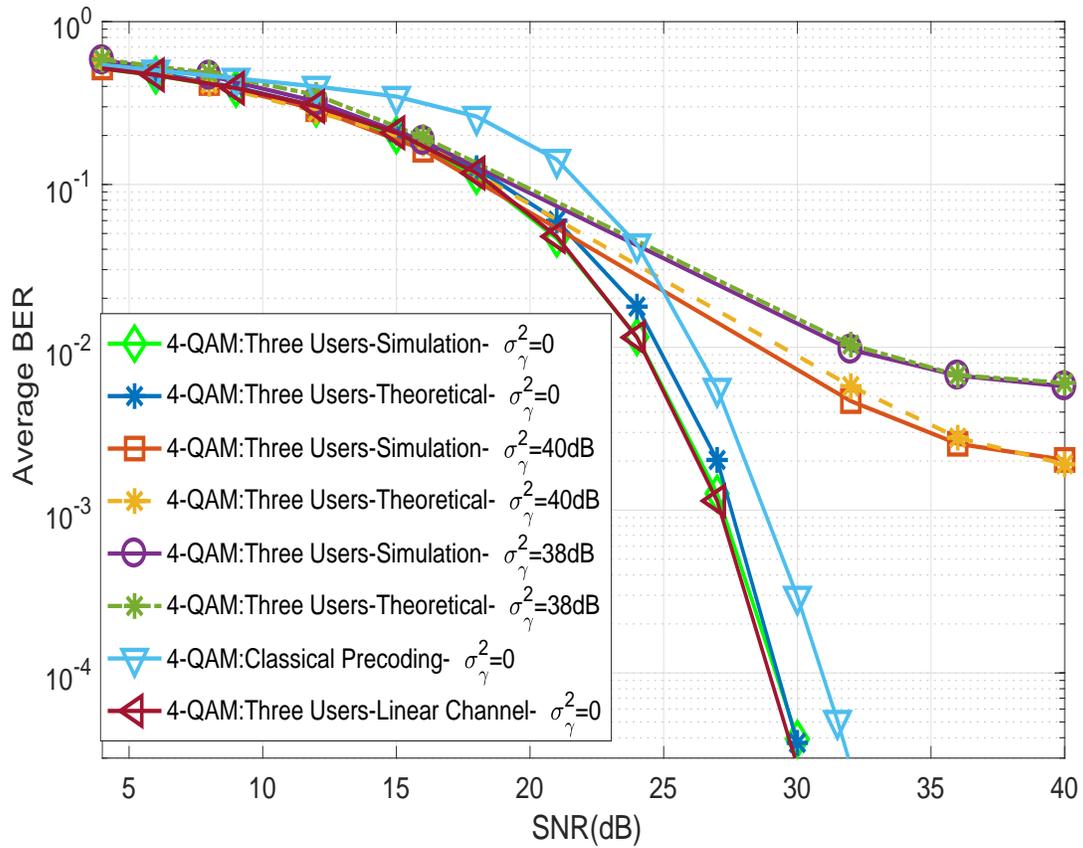}
  \caption{Plot of average BER vs SNR for the proposed approach, and classical precoding for $U=3$ for the proposed technique for various estimation error variances for 4-QAM.}\label{fig_4}
\end{figure}
\begin{figure}[!htbp]
  \centering
  \includegraphics[width=1\linewidth,height=11cm]{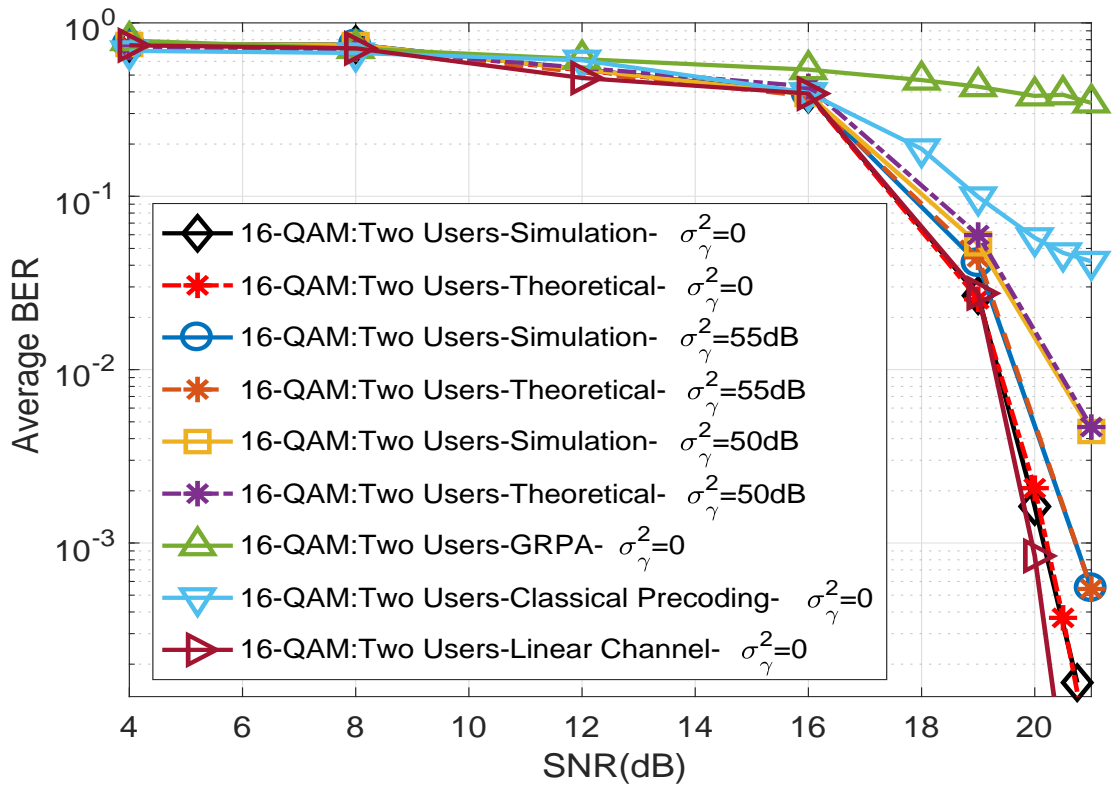}
  \caption{Plot of average BER vs SNR for the proposed approach, GRPA, and classical precoding for $U=2$ for various estimation error levels for 16-QAM.}\label{fig_5}
\end{figure}
\begin{figure}[!htbp]
  \centering
  \includegraphics[width=1\linewidth,height=11cm]{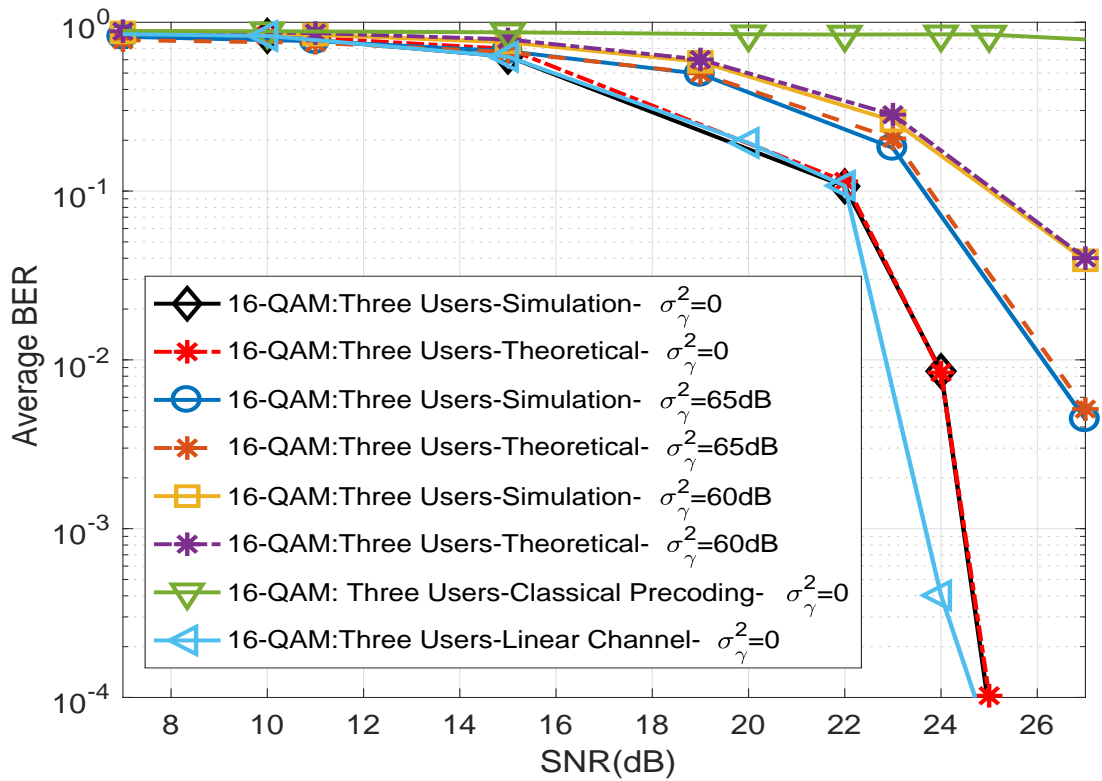}
  \caption{Plot of average BER plotted as a function of SNR for $U=3$ for the proposed approach, and the classical precoding  for various estimation error levels for 16-QAM.}\label{fig_6}
\end{figure}

\end{document}